\documentclass[runningheads]{llncs}
\usepackage[english]{babel}
\usepackage{tabularx}
\usepackage{textcomp}
\usepackage{xcolor}
\usepackage{booktabs}
\usepackage{mathrsfs}
\usepackage{hyperref}
\usepackage{multicol}
\usepackage[linesnumbered,ruled,lined,noend,algo2e]{algorithm2e}
\usepackage{algorithm,algpseudocode}
\usepackage{lipsum}
\usepackage[utf8]{inputenc}
\usepackage[T1]{fontenc}
\usepackage{amsmath}
\usepackage{amssymb}
\usepackage{graphicx}
\usepackage{multirow}
\usepackage[font=scriptsize]{caption}
\usepackage[export]{adjustbox}
\usepackage{wrapfig}
\usepackage{graphicx}

\newcommand{\C}{{\sf cost}}
\newcommand{\clr}{{\sf Color}}
\newcommand{\idN}{{\sf idNumber}}
\newcommand{\binN}{{\sf binNumber}}
\newcommand{\HeapMin}{\sf H^{min}}
\newcommand{\gidN}{{\sf gIdNumber}}
\newcommand{\tree}{\sf T}
\newcommand{\findmin}{\sf extractMin}
\newcommand{\insertH}{\sf insertHeap}
\newcommand{\deleteH}{\sf deleteHeap}
\newcommand{\insertT}{\sf insertBST}
\newcommand{\deleteT}{\sf deleteBST}
\newcommand{\optoffline}{\sf OPT_{off}}
\newcommand{\idTobin}{\sf idTobinNumber}
\newcommand{\idToCap}{\sf binToCapacity}
\newcommand{\searchmap}{\sf searchMap}
\newcommand{\insertmap}{\sf insertMap}
\newcommand{\ourAlgoFirst}{\sf DSusingBPC}
\newcommand{\gBin}{\sf gBinNumber}
\newcommand{\VarBinDS}{\sf VariableSizeDS}
\newcommand{\IntervalGen}{\sf IntervalGenerator}

\begin{document}
\title{Online Drone Scheduling for Last-mile Delivery 
}


\author{Saswata Jana\inst{1} \and Giuseppe F. Italiano\inst{2} \and Manas Jyoti Kashyop\inst{2} \and Athanasios L. Konstantinidis\inst{2} \and  Evangelos Kosinas\inst{3} \and
Partha Sarathi Mandal\inst{1}
}
\institute{Indian Institute of Technology Guwahati, Guwahati, India\\
\email{\{saswatajana, psm\}@iitg.ac.in} \and
Luiss University, Rome, Italy\\
\email{gitaliano@luiss.it}, \email{kashyopmanas@gmail.com},
\email{akonstantinidis@luiss.it}\and
University of Ioannina,  Ioannina, Greece\\
\email{ekosinas@cs.uoi.gr}}
\authorrunning{Jana et al.}
%

%
\maketitle              
\begin{abstract}
Delivering a parcel from the nearest distribution hub to the customer's doorstep is called the \textit{last-mile delivery} step in delivery logistics. Last-mile delivery is the costliest and most time-consuming phase of the entire delivery logistics. It is gaining importance due to the substantial rise in the e-commerce market and customers' expectations of quicker delivery. Delivery by \textit{unmanned aerial vehicles} (commonly known as drones) has garnered significant interest due to its environment-friendly operations, faster delivery times, low operational charges, and the ability to reach remote areas.
In this paper, we study a hybrid {\it truck-drones} model for the last-mile delivery step. A  truck moves on a path carrying drones and parcels. 
Drones do the deliveries, whereas the moving truck is used as a base for launching and landing drones synchronously.
We define the \textsc{online drone scheduling} problem, where the truck moves in a predefined path, and the customer's requests appear online during the truck's movement. The objective is to schedule a drone associated with every request to minimize the number of drones used subject to the battery budget of the drones and compatibility of the schedules.
We propose a 3-competitive deterministic algorithm using the next-fit strategy and 2.7-competitive algorithms using the first-fit strategy for the problem with $O(\log n)$ worst-case time complexity per request, where $n$ is the maximum number of active requests at any time. We also introduce \textsc{online variable-size drone scheduling} problem (OVDS). Here, we know all the customer's requests in advance; however, the drones with different battery capacities appear online. The objective is to schedule customers' requests for drones to minimize the number of drones used. We propose a $(2\alpha + 1)$-competitive algorithm for the OVDS problem with total running time $O(n \log n)$ for $n$ customer requests, where $\alpha$ is the ratio of the maximum battery capacity to the minimum battery capacity of the drones. Finally, we address how to generate intervals corresponding to each customer request when there are discrete stopping points on the truck's route, from where the drone can fly and meet with the truck.
\end{abstract}

%
%
%
\newpage
\section{Introduction}
\label{sec:Introduction}
Rapid developments in unmanned aerial vehicle (UAV) technology, popularly known as drone technology, motivate scientists and researchers to use drones in various operations. The flexibility and mobility of drones allow them to be used in healthcare, weather forecasts, surveillance, disaster management, agriculture, artificial intelligence, and many more. Recently, delivery giant companies, e.g., Amazon, DHL, and FedEx, have started using drones for the \textit{last-mile delivery} \cite{amazon}, \cite{DHL}. Last-mile delivery is the delivery logistic system's final and costliest step. In this step, the package is collected from the distribution hub and delivered to the customer using a vehicle, e.g., a delivery truck. The problem of delivering goods using vehicles is studied as the \textit{routing problem}, where the objective is to minimize the total make-span. If a single vehicle without any capacity constraint is used in the routing problem, the problem reduces to the \textit{traveling salesman problem (TSP)}. If several vehicles with capacity constraints are used, the routing problem is called the \textit{vehicle routing problem} (VRP). In literature, several variants of both problems are considered, e.g., deployment of depots for refilling the delivery trucks, bounding the number of depots, and introducing time windows to serve a customer~ \cite{laporte1992traveling}, \cite{laporte1992vehicle}, all with the final objective to minimize the total make-span.  The use of drones along with delivery vehicles was studied for the first time by Murray and Chu \cite{MURRAY201586}, where the authors proposed a model called the \textit{flying sidekicks traveling salesman problem} (FSTSP). In this model, a delivery truck and a single drone cooperate to deliver the packages. The truck starts its journey from the depot along with the deliverables and carries the drone. A customer is either served by the truck or by the drone. The authors presented a \textit{mixed integer linear programming} (MILP) formulation and heuristics for finding the best schedule using the model. Later, the authors in \cite{CRISAN201938}, \cite{Murray2019TheMF}  modified the FSTSP model and introduced multiple drones (extended FSTSP model). However, in the extended FSTSP model, the authors did not consider the limitation of the flight endurance of the drones. Daknama and Kraus \cite{Daknama2017VehicleRW} proposed a variation of the extended FSTSP model where every drone has a limited battery capacity (budgeted extended FSTSP model)  and thus returns to the truck after completing the delivery for recharging. Authors in \cite{wang2017vehicle} considered the budgeted extended FSTSP model from the worst-case point of view and gave comparative results depending on the number of drones and the relative speed of drone and truck. In \cite{ahani2020age},  \cite{ghazzai2019future},\cite{kim2018hybrid}, and \cite{park2017battery}, authors considered different UAV-based frameworks with a focus on battery recharging policy. They proposed heuristics-based approaches to optimize different objective functions, such as minimizing the total travel cost, minimizing the energy consumption, and maximizing the total revenue or coverage efficiency. 

The use of drones in delivery logistics has several advantages. Reducing $CO_2$ emissions is a major concern in the transportation system. Recent studies suggest that the use of drones in the delivery logistic system reduces $CO_2$ emissions compared to the use of only delivery trucks or human-operated ground vehicles, for example:~\cite{drone-co-2},~\cite{Co2Emissions_Goodchild}, and ~\cite{drone-co-1}. Using drones in the delivery logistics system has other motivations, such as avoiding traffic congestion by flying,  following a shorter route toward the customer's doorstep, and contactless delivery. Delivery with minimal human involvement or contactless delivery became crucial after the recent pandemic. However, using only drones has certain limitations, such as the size of the deliverable packages and the distance of the customer from the launching location due to limited battery capacity. Therefore, a hybrid model with a delivery truck and several drones (truck-drones model) is more desirable~\cite{wang2017vehicle}. In the truck-drones model, the truck containing all the drones and the packages follows conventional ground routes, commencing its journey from the warehouse. A customer will be delivered either by truck or by drone. In \cite{Boysen2018DroneDF}, authors considered that the truck has some specific discrete stopping points where the drone can take off from the truck and land on the truck. A customer places a request with information about its location, and based on that, a launching point and a rendezvous point for the drone to be used to serve the request are computed. Several works (e.g., \cite{jaillet2008online}) considered that customer requests will appear online, and the objective is to minimize the total make-span. Sorbeli et al. \cite{Betti}
proposed a similar model in the offline version, where all the customer positions and the corresponding launching and meeting points are known in advance. The objective is to maximize the profit by serving a subset of customers with a fixed number of drones subject to battery and compatibility constraints. Authors in \cite{janaPacking} took a variant of the above model where the truck has a sufficient number of drones to complete all the deliveries, and the objective is to use the minimum number of drones.  The authors in \cite{Betti}, \cite{janaPacking}, and \cite{janaSchedulingJournal} proposed several heuristics and approximation algorithms for this offline version. \par
In this work, we follow the hybrid approach and consider the truck-drones model. Our objective is to use the minimum number of drones. For simplicity, we do not mention the customer requests that are close to the truck path and are served by the truck. Further, we consider that the truck has some specific discrete stopping points where the drone can take off from the truck and land on the truck. We compute the launch and rendezvous point of each customer request for a drone, and these requests appear online. Based on the launch and rendezvous point, an interval generator generates an interval for every customer request. We call the amount of battery consumption of a drone used to serve a specific request as the cost of the interval corresponding to that request. The objective is to use the minimum number of drones to serve all the requests. Thus, we differ from the previous works that considered online customer requests intending to minimize the total make-span. We assumed that the truck moves on a known path containing a sufficient number of drones, and the path has some discrete stops from where the drone can launch or meet the truck. Also, the time to reach or leave those points by truck is known. Now a request from a customer will come at any point on the path. We need to select the best points for launching and rendezvous so that the cost to complete the delivery is minimum. 
\subsection{Our Contribution}
\begin{itemize}
    \item We define the \textsc{online drone scheduling} problem to minimize the number of drones in the last-mile delivery scenario.
    \item We propose a deterministic 3-competitive algorithm for the \textsc{online drone scheduling} problem by using the next-fit bin packing strategy.
    \item We propose a 2.7-competitive algorithm for the problem by
    applying the first-fit bin packing strategy.
    \item We show that all the above algorithms use $O(\log n)$ worst-case time complexity per scheduling, where $n$ is the maximum number of active requests at any time. We define some data structures to achieve this time efficiency.
    \item We define another variant of the problem named \textsc{online variable-size drone scheduling} and propose an $(2\alpha +1)$- competitive algorithm with time complexity $O(n \log n)$, where $\alpha$ is the ratio between the maximum and minimum battery capacity of the drones.
    \item We also discuss how to generate delivery intervals in an online setting having some discrete stop points on the truck path.
\end{itemize}
\subsection{Preliminaries}
\label{sec:Prelim}
A graph is an \emph{interval graph} if there is a bijection between its vertices and a family of closed intervals of the real line such that two vertices are adjacent if and only if the two corresponding intervals intersect. 
Such a bijection is called an \emph{interval representation} of the graph, denoted by $\mathcal{I}$.
We identify the intervals of the given representation with the vertices of the graph, interchanging these notions appropriately.
Whether a given graph is an interval graph can be decided in linear time and if so, an interval representation can be generated in linear time~\cite{FG65}.
Notice that every induced subgraph of an interval graph is an interval graph.
A {\it clique} of $G$ is a set of pairwise adjacent vertices of $G$, and a {\it maximal clique} of $G$ is a clique of $G$ that is not properly contained in any clique of $G$. An {\it independent set} of $G$ is a set of pairwise non-adjacent vertices of $G$.

\noindent \textbf{Organization:} The rest of the paper is organized as follows. Section \ref{sec:DroneScehedulingBPC} introduces the problem definition for the \textsc{online drone scheduling} problem and a few algorithms to solve it.
Section \ref{sec:VariableSizeDS} defines the \textsc{online variable-size drone scheduling} problem along with an algorithm. Section \ref{sec:IntervalGen} discusses interval generation. Finally we conclude in Section \ref{sec:conclusion}.

\section{Drone scheduling through bin packing with conflicts}
\label{sec:DroneScehedulingBPC}

We consider the interval representation of the underlying interval graph $G$. Let $I_t = [l_t, r_t]$ be the interval inserted at time $t$. Let $v_t$  be the vertex in $G$ corresponding to $I_t$. Therefore, the insertion or deletion of an interval is the insertion or deletion of the corresponding vertex of the underlying graph $G$. In our work, every 
interval $I_t$ is also associated with a cost ($\C(I_t)$) that corresponds to the amount of battery consumption of a drone while serving the customer request corresponding to the interval.
In this section, we consider the \textsc{online drone scheduling} problem as proper coloring of intervals with a fixed budget $B$ for every color, where $B$ is a positive integer and is part of the input. A color assigned to an interval corresponds to a drone assigned to serve the request represented by the interval and budget $B$ of the color corresponds to the battery capacity of the drone. 


\noindent
We define the following problem.
 \begin{definition}
\label{def:DroneScheduling-as-BPC}	
{\bf (\textsc{Online Drone Scheduling})}
Intervals are appearing online. Suppose at time step $t$, interval $I_t = [l_t, r_t]$ appears with cost $\C(I_t)$ with $l_t \geq t$. The goal is to assign a color to the interval with the following constraints:\\
$1$. Any two Intervals with the same color are disjoint (non-overlapping).\\
$2$. For a particular color, the summation of the costs of all intervals with that color is at most $B$. 
\end{definition}
To solve the problem we will use the algorithms developed for online bin packing with conflicts  in~\cite{EpsteinL-BPC}. For interval $I_t$ at time $t$, our goal is to compute a color $\clr(I_t)$.
A color $\clr(I_t)$ is a $2$-tuple, $\clr(I_t)$ = $(\idN(I_t), \binN(I_t))$, with the following properties.
\begin{enumerate}
\item ({\bf conflict free}) No two overlapping intervals have the same $\idN$. Thus, the $\idN$ ensures a conflict free property.
\item ({\bf budget constraint}) For all the intervals with the same  $\idN$ and $\binN$ (in other words the same color), the summation of their costs must be at most $B$. Therefore, for the intervals with the same $\idN$, $\binN$ ensures that the budget constraint associated with every color is satisfied.
\item ({\bf proper coloring}) The above two properties imply that for any two intervals $I_t \neq I_{t^{\prime}}$ and $I_t \cap I_{t^{\prime}} \neq \emptyset$,
$(\idN(I_t), \binN(I_t)) \neq (\idN(I_{t^{\prime}}),$ $ \binN(I_{t^{\prime}}))$.\\
\end{enumerate}
Based on the above properties, we use the following to compute the color.
\begin{enumerate}
\item $\idN$ is computed using an online interval coloring algorithm where the colors are represented using positive integer numbers. Note that the coloring algorithm used at this step does not have any budget associated with the colors and thus can be assigned to any number of non-overlapping intervals.
For an interval $I_t$, the color (the non-negative integer) allotted to $I_t$ at this step using the online coloring algorithms is assigned as the $\idN(I_t)$.
\item Once the $\idN(I_t)$ is computed for an interval $I_t$ in the previous step, next step is to compute $\binN(I_t)$ ensuring that all the intervals with the same color as $(\idN(I_t), \binN(I_t))$ satisfies the budget constraint $B$. Therefore, we use an online bin packing algorithm to compute the $\binN$ for $I_t$ considering all the previously inserted intervals with the $\idN$ equal to $\idN(I_t)$.\\
\end{enumerate}
For an interval $I_{t^{\prime}} = [l_{t^{\prime}}, r_{t^{\prime}}]$, color of the interval is computed at time $t = l_{t^{\prime}}$. Further, we assume that intervals have distinct endpoints. Therefore, at the time step $t$, the following scenarios occur while handling the update. 
\begin{enumerate}
\item If a new interval $I_t$ appears, then $I_t$ is inserted. Further, at time step $t$ at most one interval appears. 
\item If there is an existing interval  $I_{t^{\prime}}$ with $t^{\prime} < t$ and $r_{t^{\prime}} = t$, then  $I_{t^{\prime}}$ is deleted. 
\item If there is an interval $I_{t^{\prime}}$ with $t^{\prime} \leq t$ and $l_{t^{\prime}} = t$, then color for $I_{t^{\prime}}$ is computed.
\item As intervals have distinct endpoints, scenario $2$ and $3$ are mutually exclusive.\\
\end{enumerate}
Because of the properties of our original \textsc{online drone scheduling} problem, the intervals appearing online satisfy the following properties.
\begin{itemize}
\item Interval appearing at time step $t$ has left endpoint at least $t$. \item Our deletion scheme is that at time step $t$, if there is an interval $I_{t^{\prime\prime}}$ with $t^{\prime\prime} < t$ and $r_{t^{\prime\prime}} = t$, then we delete $I_{t^{\prime\prime}}$. Otherwise, if there is an interval $I_{t^{\prime}}$ with $t^{\prime} \leq t$ and $l_{t^{\prime}} = t$, the color for $I_{t^{\prime}}$ is computed. Note that because of our deletion scheme, while computing the color for $I_{t^{\prime}}$, all the intervals overlapping with $I_{t^{\prime}}$ must contain the point $t$ forming a clique. Further, because of our coloring scheme, colors are assigned to the intervals in the increasing order of their left endpoints.
\end{itemize}
\subsection{Data structures}
\label{sec:DataStruct-BPC-intervals}
We use the following data structures.
\begin{itemize}
\item We store every interval $I_t = [l_t, r_t]$ with $\C(I_t)$ in a map $\mathcal{I}$ using $t$ as key.
\item The left endpoint $l_t$ for every interval $I_t = [l_t, r_t]$ is stored as $(l_t, t)$ in a  min heap $\HeapMin_{L}$ using $l_t$ as key.
\item The right endpoint $r_t$ for every interval $I_t = [l_t, r_t]$ is stored as $(r_t, t)$ in a  min heap $\HeapMin_{R}$ using $r_t$ as key.
\item Color of every interval $I_t = [l_t, r_t]$ is stored as $(\idN(I_t), \binN(I_t))$ in a map $\clr$ with $t$ as the key. In the case, interval $I_t$ is deleted or the color for $I_t$ is yet to be computed, $\clr[t]$ stores NULL.
\end{itemize}
The operation $\findmin()$ returns the minimum element in a min heap. The operation $\insertH()$ inserts an element into a heap and $\deleteH()$ deletes the root element of the heap (that is the minimum element in the case of min heap and the maximum element in the case of max heap). $\searchmap()$ is used to search for a key in a map and $\insertmap()$ is used to insert an entry into a map. In addition, we will use some data structures specific to the computation of $\idN$ and $\binN$, and we will describe them in the relevant sections. 
\subsection{Computing $\idN$}
\label{sec:ComputeIdnumber}
At time step $t$, we consider the interval $I_{t^{\prime}} = [l_{t^{\prime}}, r_{t^{\prime}}]$ with $l_{t^{\prime}} = t$ and $t^{\prime} \leq t$ to compute $\idN(I_{t^{\prime}})$.
Let $\mathcal{I}_t$ be the set of intervals overlapping with $I_{t^{\prime}}$ and contain the point $t$. Since all the intervals have distinct endpoints, the left endpoint of every interval in $\mathcal{I}_t$ is strictly smaller than $t$. We observe the following properties about the set of intervals $\mathcal{I}_t$.
\begin{enumerate}
\item $\mathcal{I}_t$ forms a clique because all the intervals in $\mathcal{I}_t$ contains the point $t$.
\item  Because of our deletion scheme, $\mathcal{I}_t$ is the only set of intervals overlapping with the interval $I_{t^{\prime}}$ and with left endpoints strictly less than $t$.
\item As of our coloring scheme, color for each interval in $\mathcal{I}_t$ is already computed. 
\item Any interval not in $\mathcal{I}_t$ and overlapping with $I_{t^{\prime}}$ must have left endpoint strictly greater than $t$. Because of our coloring scheme, the color for such an interval is not yet computed and we do not consider such an interval while computing the color for $I_{t^{\prime}}$ in the current step.
\end{enumerate} 
Our goal is to find the smallest $\idN$ not assigned to any interval in $\mathcal{I}_t$ and assign it as $\idN(I_{t^{\prime}})$. Along with the data structures described in Section~\ref*{sec:DataStruct-BPC-intervals}, we use the following data structures specifically to compute $\idN$.
\begin{itemize}
\item A counter $\gidN$ that stores the maximum value of $\idN$ used so far. It is initialized to $0$. 
\item A min heap $\HeapMin_{F\idN}$ (a heap of free $\idN$s) containing the numbers in the range $\{1,\cdots,\gidN\}$ that are not used by any existing interval as $\idN$. 
\item A balanced binary search tree $\tree_{U\idN}$ containing the numbers in the range $\{1,\cdots,\gidN\}$ that are used by an existing interval as $\idN$. The operation $\insertT()$ inserts an element to the tree and $\deleteT()$ deletes an element from the tree.
\end{itemize}
  Note that because of our deletion scheme, the only $\idN$s in use are the ones corresponding to the colors of the intervals in $\mathcal{I}_t$. Therefore, the tree $\tree_{U\idN}$ stores only the $\idN$s of the intervals in $\mathcal{I}_t$ maintaining the exact set of $\idN$s that are forbidden for us to use. Similarly, the heap $\HeapMin_{F\idN}$ maintains all the free $\idN$s that are potential candidates for us to use. Using the above data structures, we design the procedure Algorithm~\ref{alg:BPC-computeidNumber}.
\begin{algorithm2e}[h]
	\caption{compute-idNumber$(I_{t^{\prime}})$}	
	\label{alg:BPC-computeidNumber}
            \DontPrintSemicolon
		\If{$\HeapMin_{F\idN}$ is Empty}{
		$\gidN \leftarrow \gidN + 1$\;
		$\idN(I_{t^{\prime}})  \leftarrow \gidN$\;
            }
		\Else{
		$\idN(I_{t^{\prime}})  \leftarrow \findmin(\HeapMin_{F\idN})$\;
		$\deleteH(\HeapMin_{F\idN})$\;
		}
        $\insertT(\tree_{U\idN}, \idN(I_{t^{\prime}}))$\;
\end{algorithm2e}
\begin{lemma}
\label{lem:correctness-of-computeIdN}	
{\sf compute-idNumber()} compute $\idN$ for each interval $I_{t^{\prime}}$.
\end{lemma}
\begin{proof}
    Procedure {\sf compute-idNumber()} (refer to Algorithm~\ref{alg:BPC-computeidNumber}), first checks if there is a previously used $\idN$ available for reuse. In the case of such availability, it selects the minimum one and assigns it as the  $\idN$ of the current interval, and adjusts the relevant data structures suitably  (Line $4-6$ in Algorithm~\ref{alg:BPC-computeidNumber}). Otherwise, all the intervals overlapping with the current interval and excluding the current interval form a clique of size $\gidN$. Therefore, a new $\idN$ must be introduced to assign it as the $\idN$ of the current interval (Lines $1-3$ in Algorithm~\ref{alg:BPC-computeidNumber}). 
    Therefore, the procedure {\sf compute-idNumber()} always assigns to the current interval the least $\idN$ not used by any of the existing intervals overlapping with it.\newline
\end{proof}
\begin{lemma}
\label{lem:runningtime-of-computerIdN}	
Running time of procedure {\sf compute-idNumber()} is $O(\log n)$ in the worst-case.
\end{lemma}
\begin{proof}
    The running time of procedure {\sf compute-idNumber()} is dominated by the standard operations in min heap and binary search tree. Therefore, in the worst-case the procedure takes $O(\log n)$ time, where $n$ is the maximum number of elements present in the heap or the binary search tree throughout the handling of the online interval sequence.\newline 
\end{proof}
\begin{lemma}
\label{lem:gidNLeqOpt}	
$\gidN \leq \optoffline$, where $\optoffline$ is the number of colors used by an optimal offline interval coloring algorithm without any budget constraints on the colors.
\end{lemma}
\begin{proof}
The optimal offline interval coloring uses a greedy strategy. It considers the intervals in the sorted order of their left endpoints and assigns an interval to the least numbered color not already assigned to any of its neighbors. Notice that the procedure {\sf compute-idNumber()} mimics the offline greedy strategy. Further, because of our deletion scheme, the set of free $\idN$s available is a super set of the set of available colors in the offline greedy strategy. Since $\gidN$ denotes the maximum value of $\idN$ used so far by the algorithm, we have $\gidN \leq \optoffline$.
\end{proof}
\subsection{Computing $\binN$}
\label{sec:Next-fit}
At time step $t$, after computing  $\idN(I_{t^{\prime}})$ for the interval $I_{t^{\prime}} = [l_{t^{\prime}}, r_{t^{\prime}}]$ with $l_{t^{\prime}} = t$ and $t^{\prime} \leq t$,  the next task is to compute $\binN$ for $I_{t^{\prime}}$. First we use the next-fit strategy from online bin packing to compute the $\binN$.
\begin{definition}
\label{def:next-fit}
({\bf next-fit})
The next-fit algorithm for online bin packing maintains an active bin. Initially, the active bin is empty. When a new item appears, if the item fits into the active bin, then the next-fit algorithm places the item in the active bin. Otherwise, the current active bin is closed, a new bin is marked as active, and the current item is placed in the new active bin.  
\end{definition}
To use the next-fit strategy for our purpose,  we maintain an active bin number corresponding to every $\idN$ along with its capacity (precisely, remaining capacity). Once an item with a specific $\idN$ does not fit into the current active bin number, a new bin is made active with capacity $B$. We use the following data structures to compute $\binN$ using next-fit.
\begin{itemize}
\item A map $\idTobin$ with $\idN$ as key which stores the currently active bin number corresponding to the  $\idN$. For every new $\idN$, the entry $\idTobin[\idN]$ is initialized to $1$ (that is the bins corresponding to a particular $\idN$ are numbered starting from $1$). 
\item  A map $\idToCap$ with $\idN$ as the key which stores the remaining capacity of the currently active bin number corresponding to the  $\idN$. For every new $\idN$, the entry $\idToCap[\idN]$ is set to $B$.
\end{itemize}
At time step $t$, using the above data structures, we design the following procedure (Algorithm~\ref{alg:BPC-computeBinNumber-nextfit}) to compute $\binN$ for an interval $I_{t^{\prime}} = [l_{t^{\prime}}, r_{t^{\prime}}]$ with $t^{\prime} \leq t$ and $l_{t^{\prime}} = t$ after computing $\idN(I_{t^{\prime}})$.
{\small
\begin{algorithm2e}[h]
	\caption{compute-binNumber-Nextfit$(I_{t^{\prime}})$}	
	\label{alg:BPC-computeBinNumber-nextfit}
        \DontPrintSemicolon	 
		\If{$\searchmap(\idTobin, \idN(I_{t^{\prime}}))$ is Empty}{
		$\insertmap(\idTobin, <\idN(I_{t^{\prime}}),1>)$\;
		$\insertmap(\idToCap, <\idN(I_{t^{\prime}}), B - \C(I_{t^{\prime}})>)$\;
         }
		\Else{
		\uIf{$\idToCap[\idN(I_{t^{\prime}})] \geq \C(I_{t^{\prime}})$}{
		$\idToCap[\idN(I_{t^{\prime}})] \leftarrow \idToCap[\idN(I_{t^{\prime}})] - \C(I_{t^{\prime}})$\;
  }
		\Else{
		$\idTobin[\idN(I_{t^{\prime}})] \leftarrow \idTobin[\idN(I_{t^{\prime}})] + 1$\;
		$\idToCap[\idN(I_{t^{\prime}})] \leftarrow B - \C(I_{t^{\prime}})$\;
		}
		}
		$\binN(I_{t^{\prime}}) \leftarrow \idTobin[\idN(I_{t^{\prime}})]$\;
\end{algorithm2e}}
\begin{lemma}
\label{lem:correctness-of-computeBinN-nextfit}	
{\sf compute-binNumber-Nextfit()} computes $\binN$ for each $I_{t^{\prime}}$.
\end{lemma}
\begin{proof}
    For the interval $I_{t^{\prime}}$, the procedure {\sf compute-binNumber-Nextfit()} (refer to Algorithm~\ref{alg:BPC-computeBinNumber-nextfit}) first checks if there is any active bin corresponding to $\idN(I_{t^{\prime}})$. In the case when there is no active bin for $\idN(I_{t^{\prime}})$, the first active bin with number $1$ is created, the remaining capacity of the new bin is adjusted suitably, and  $\binN(I_{t^{\prime}})$ is assigned the value $1$ (Line $1-3$ and Line $10$). Otherwise, if the current active bin has sufficient remaining capacity to accommodate interval $I_{t^{\prime}}$, then the remaining capacity of the bin is adjusted suitably, and  $\binN(I_{t^{\prime}})$ is assigned the number of the current active bin (Line $5-6$ and Line $10$). In the case when the remaining capacity of the current active bin is insufficient to accommodate  $I_{t^{\prime}}$, a new bin corresponding to $\idN(I_{t^{\prime}})$ is created, and is assigned as $\binN(I_{t^{\prime}})$ after appropriately adjusting the remaining capacity (Line $7-9$ and Line $10$). Therefore, procedure {\sf compute-binNumber-Nextfit()} mimics the next-fit strategy ensuring the constrain that sum of the costs of all intervals with color $(\idN(I_{t^{\prime}}), \binN(I_{t^{\prime}}))$ is at most $B$.\newline
\end{proof}
\begin{lemma}
\label{lem:runningtime-of-computeBinN-nextfit}	
The running time of procedure  {\sf compute-binNumber-Nextfit()} is $O(\log n)$ in the worst-case.
\end{lemma}
\begin{proof}
    The running time of procedure {\sf compute-binNumber-Nextfit()} (Algorithm~\ref{alg:BPC-computeBinNumber-nextfit}) is dominated by the standard search ($\searchmap()$) and insert ($\insertmap()$) operation in a map data structure. The operation  $\searchmap()$ takes $O(1)$ time in the worst-case. The operation $\insertmap()$ takes $O(\log n)$ time in the worst-case where $n$ is the maximum number of elements present in the map data structure throughout handling the online interval sequence. Therefore, the worst-case running time of procedure {\sf compute-binNumber-Nextfit()} is $O(\log n)$.\newline
\end{proof}
Finally, we present the algorithm $\ourAlgoFirst$ (online drone scheduling using online bin packing with conflicts, refer to Algorithm~\ref{alg:DSusingBPC} for pseudo code) which handles the update at time $t$ as follows.
\begin{itemize}
\item If a new interval $I_t$ is to be inserted, then insertion is performed by updating the relevant map and heap data structures (Line $1-4$ in Algorithm~\ref{alg:DSusingBPC}).
\item If there is an interval $I_{t^{\prime\prime}}$ with $r_{t^{\prime\prime}} = t$, then $I_{t^{\prime\prime}}$ is deleted by updating the relevant data structures (Line $5-9$ in Algorithm~\ref{alg:DSusingBPC}).
\item If  an interval $I_{t^{\prime\prime}}$ with $l_{t^{\prime\prime}} = t$, then color for $I_{t^{\prime\prime}}$ is computed by calling Algorithm~\ref{alg:BPC-computeidNumber} (Line $13$ in Algorithm~\ref{alg:DSusingBPC}) and Algorithm~\ref{alg:BPC-computeBinNumber-nextfit} (Line $14$ in Algorithm~\ref{alg:DSusingBPC}).
\end{itemize}

\noindent In the next lemma we prove that $\ourAlgoFirst$ is $3$-competitive similar as of \cite{EpsteinL-BPC}.
\begin{algorithm2e}[H]
\caption{$\ourAlgoFirst(I_{t} = [l_{t}, r_{t}])$}	
\label{alg:DSusingBPC}
\DontPrintSemicolon
\If{$I_{t} \neq NULL$}{
$\insertmap(\mathcal{I}, < t, I_{t} >)$\;
$\insertH(\HeapMin_{L}, (l_{t}, t) )$\;
$\insertH(\HeapMin_{R}, (r_{t}, t) )$\;
}
$r_{t^{\prime\prime}} \leftarrow \findmin(\HeapMin_{R})$\;
\If{$r_{t^{\prime\prime}} == t$ }{
$\insertH(\HeapMin_{F\idN}, \idN(I_{t^{\prime\prime}}))$\;
$\deleteT(\tree_{U\idN}, \idN(I_{t^{\prime\prime}}))$\;
$\deleteH(\HeapMin_{R})$\;
}
\Else{
$l_{t^{\prime\prime}} \leftarrow \findmin(\HeapMin_{L})$\;
\If{$l_{t^{\prime\prime}} == t$}{
compute-idNumber$\sf (I_{t^{\prime\prime}})$\;
compute-binNumber-Nextfit$\sf (I_{t^{\prime\prime}})$\;
$\insertmap(\clr, <t^{\prime\prime},$ $ (\idN(l_{t^{\prime\prime}}),$ $ \binN(l_{t^{\prime\prime}}))>)$\;
$\deleteH(\HeapMin_{L})$\;
}
}
\end{algorithm2e}

\begin{lemma}
\label{lem:Ourresult-3-competitve}
$\ourAlgoFirst$ is $3$-competitive.
\end{lemma}
\begin{proof}
Let $\sf L_i$ be the number of bins, each of capacity $B$, created corresponding to $\idN = i$. Let $\C(L_i)$ correspond to the summation of the costs of all the intervals whose color consists of  $\idN = i$ and $\binN = L_i$. Notice that in the next-fit strategy, for any two consecutive bins, the sum of the occupied capacity in the bins is strictly greater than $B$. Therefore, it is easy to observe that $\frac{\C(L_i)}{B} \geq \frac{L_i -1}{2}$ or $L_i \leq 2\cdot \frac{\C(L_i)}{B} + 1$. Let ${\sf ALG}$ be the total number of different colors used by $\ourAlgoFirst$. Let $\optoffline^{B}$ be the number of different colors used by an optimal offline algorithm with budget constraint $B$ for every color. 
Then,
\setlength{\belowdisplayskip}{0.5pt} \setlength{\belowdisplayshortskip}{0.8pt}
\setlength{\abovedisplayskip}{0.5pt} \setlength{\abovedisplayshortskip}{0.5pt} 
\begin{equation*}
  {\sf ALG} = \sum_{i=1}^{\gidN}{L_i}  \leq 2 \cdot {\left( \sum_{i=1}^{\gidN}{\frac{\C(L_i)}{B}} \right)} + \gidN \\
\end{equation*}
It is easy to observe that $\optoffline^{B} \geq \sum_{i=1}^{\gidN}{\frac{\C(L_i)}{B}}$. In Lemma~\ref{lem:gidNLeqOpt}, we proved that $\gidN \leq \optoffline$, where $\optoffline$ is the number of different colors used by an optimal offline coloring algorithm without any budget constraints on the colors. Again, it is easy to observe that $\optoffline \leq \optoffline^{B}$. Therefore,
\begin{equation*}
    {\sf ALG} \leq 2 \cdot {\left( \sum_{i=1}^{\gidN}{\frac{\C(L_i)}{B}} \right)} + \gidN \leq  2 \cdot \optoffline^{B} + \optoffline \leq 3 \cdot \optoffline^{B}\\
\end{equation*}
Hence, algorithm $\ourAlgoFirst$ is $3$-competitive.
\end{proof}
  \begin{lemma}
\label{lem:Ourresult-updatetime}
The update time of $\ourAlgoFirst$ is $O(\log n)$ in the worst-case.
\end{lemma}
\begin{proof}
    Update time of $\ourAlgoFirst$ is dominated by standard operations in the heap, binary search tree and map, every operation taking $O(\log n)$ time in the worst-case. From Lemma~\ref{lem:runningtime-of-computerIdN} and Lemma~\ref{lem:runningtime-of-computeBinN-nextfit}, computing color for an interval takes $O(\log n)$ time in the worst-case. Therefore, update time $\ourAlgoFirst$ is $O(\log n)$ in the worst-case.
\end{proof}
Using Lemma~\ref{lem:Ourresult-3-competitve} and Lemma~\ref{lem:Ourresult-updatetime},  we get the following result.
\begin{theorem}
\label{thm:3competitiveDS}
There exists a deterministic $3$-competitive algorithm for online \textsc{drone scheduling} problem that performs every drone allotment in $O(\log n)$ worst-case time, where $n$ is the maximum number of requests being served by some drone at any instance (active requests) through out the handling of the request sequence. 
\end{theorem}
\subsection{Improved algorithm using first-fit for computing $\binN$}
\label{sec:first-fit}
In this section, we show that the competitive ratio of our result in Theorem~\ref{thm:3competitiveDS} can be improved by using the first-fit strategy from online bin packing to compute $\binN$ for an interval.
\begin{definition}
\label{def:first-fit}
({\bf first-fit})
The first-fit algorithm for online bin packing maintains a collection of active bins in sorted order (increasing order) with respect to their appearance. When a new item appears, among all the active bins that can accommodate the new item, it finds the first bin where the item can fit.  If no such bin exists, a new active bin is created, and the item is placed in the new bin. When the remaining capacity of an active bin becomes zero, the bin is closed and marked as inactive.
\end{definition}
\begin{algorithm2e}[h]
	\caption{compute-binNumber-Firstfit$(root, I_t)$}	
	\label{alg:BPC-computeBinNumber-firstfit}
        \DontPrintSemicolon	
	\If{$root.max < \C(I_t)$}{
        $\gBin[\idN(I_t)] \gets \gBin[\idN(I_t)] +1$\;
        $insertBST(\tree_{\binN}, \gBin[\idN(I_t)], B-\C(I_t))$\;
        \textbf{return} $\gBin[\idN(I_t)]$\;
        }
        \Else{
        \uIf{$root.rem \geq \C(I_t)$}{
        \uIf{$root.left.max < \C(I_t)$ or $root.left = NULL$}{
        $decreaseKey(\tree_{\binN}, root, \C(I_t))$\;
        \textbf{return} $root.bin$\;
        }
        \Else{
        \textbf{return} compute-binNumber-Firstfit$(root.left, I_t)$\;
        }
        }
        \Else{
        \uIf{$root.left \neq NULL$ and $root.left.max \geq \C(I_t)$}{
        \textbf{return} compute-binNumber-Firstfit$(root.left, I_t)$\;
        }
        \Else{
        \textbf{return} compute-binNumber-Firstfit$(root.right, I_t)$\;
        }
        }
        }
	
\end{algorithm2e}
To implement the first-fit strategy to compute $\binN$, for every $\idN$  in the map $\idTobin$, we maintain the set of active bin numbers in a balanced binary search tree $\tree_{\binN}$ with the bin number as the key. We store another map $\gBin$ with $\idN$ as key to store the number of bins created corresponding to this $\idN$. Each $node$ in the tree $\tree_{\binN}$ has five attributes, (1) $node.bin$: bin number (2) $node.rem$: remaining capacity of the drone (3) $node.max$: maximum remaining capacity among all the nodes in the sub-tree rooted at $node$ (4) $node.left$: left child of the $node$ (5) $node.right$: right child of the node.
$insertBST(T,i,c)$ is used to insert a $node$ in the tree $T$ with the key $i$ and $node.rem = c$. $decreaseKey(T,node,c)$ is used to decrease the $node.rem$ by $c$ in the tree $T$. For both operations, we also need to update the $max$ attribute of each node on the path of $root$ to  $node$.

Suppose we need to compute $\binN$ for interval $I_t$ with cost $\C(I_t)$ and $\idN(I_t)$. 
Let $\tree_{\binN}$ be the tree stored in the map $\idTobin$ corresponding to $\idN(I_t)$. If  $\tree_{\binN}$ is empty, then a new tree is created consisting of a single $node$ with $node.bin = 1$, $node.rem = node.max = B-\C(I_t)$, and both the children set to $NULL$. $node$ set as the root of this tree and $\gBin[\idN(I_t)] = 1$. Otherwise, Algorithm \ref{alg:BPC-computeBinNumber-firstfit} is called with the $root$ of $\tree_{\binN}$ and the interval $I_t$ as its' parameter. The returned value is set as the bin number of $I_t$ i.e., $\binN(I_t)$. If $root.max < cost(I_t)$, then all the nodes have remaining battery capacity lesser than $\C(I_t)$. For this case, we need to introduce a new bin by increasing $\gBin[\idN(I_t)]$ by one and set it as the bin number of $I_t$. A $node$ corresponding to this newly introduced bin needs to be added to the tree with $node.bin = \gBin[\idN(I_t)]$ and $node.rem = B-\C(I_t)$ (Line $1-4$). Otherwise ( $root.max \geq \C(I_t)$), we need to find a first bin to fit the interval $I_t$. If $root.rem \geq \C(I_t)$ (Line $6$), then either the root node or a node in its' left subtree is the first node to accommodate $I_t$. If the maximum remaining capacity of the roots' left subtree is lesser than $\C(I_t)$ (Line $7$), then the bin number to the root node is set as the bin number of $I_t$ (Line $9$). Also, we need to decrease $root.rem$ by $\C(I_t)$
(Line $8$). If $root.left.max \geq \C(I_t)$, then we need to find the bin into the left subtree of the root by calling the Algorithm \ref{alg:BPC-computeBinNumber-firstfit} for the $root.left$ node (Line $10-11$ and $13-14$). If $root.rem < \C(I_t)$ and $root.left.max < \C(I_t)$, then the bin to fit the interval $I_t$ must be in the right subtree of root. So we call the Algorithm \ref{alg:BPC-computeBinNumber-firstfit} for this case as $root.right$ as its' parameter (Line $15-16$). All these operations are standard operations in a balanced binary search tree and take $O(\log n)$ time in the worst-case. Therefore, the update time of the $\ourAlgoFirst$ (Algorithm~\ref{alg:DSusingBPC}) is $O(\log n)$ in the worst-case if the first-fit strategy is used to compute $\binN$ in Algorithm~\ref{alg:BPC-computeBinNumber-nextfit} instead of next-fit. Here, $n$ is the number of requests received at any instance.
\begin{theorem}
Algorithm $\ourAlgoFirst$ is $2.7$-competitive when the first-fit strategy is used to compute $\binN$.
\label{thm:first-fit}
\end{theorem}
\begin{proof}
Let $\sf L_i$ be the number of bins, each of capacity $B$, created corresponding to $\idN = i$. We apply the same weight function described in \cite{FF-analysis} for analyzing the first fit strategy on $L_i$.  After applying the weight function, let $\sf W(i)$ correspond to the summation of the weights of all the intervals whose color consists of  $\idN = i$ and $\binN$ is one of the $\sf L_i$ bins created corresponding to $\idN = i$. It is shown in \cite{FF-analysis} that $L_i \leq W(i) + 1$. Let ${\sf ALG_f}$ be the total number of different colors used by algorithm $\ourAlgoFirst$ when first-fit is used to compute $\binN$.
Let $\optoffline^{B}$ be the number of different colors used by an optimal offline algorithm with budget constraint $B$ for every color. Thus,
\setlength{\belowdisplayskip}{0.5pt} \setlength{\belowdisplayshortskip}{0.5pt}
\setlength{\abovedisplayskip}{0.5pt} \setlength{\abovedisplayshortskip}{0.5pt}  
\begin{equation*}
    {\sf ALG_f} = \sum_{i=1}^{\gidN}{L_i}  \leq \sum_{i=1}^{\gidN}{\left(W(i) + 1 \right)} = \sum_{i=1}^{\gidN}{W(i)} + \gidN\\
\end{equation*}
From the first-fit analysis of \cite{FF-analysis},  $\sum_{i=1}^{\gidN}{W(i)} \leq 1.7 \cdot \optoffline^{B}$. In Lemma~\ref{lem:gidNLeqOpt}, we proved that $\gidN \leq \optoffline$, where $\optoffline$ is the number of different colors used by an optimal offline coloring algorithm without any budget constraints on the colors. Again, it is easy to observe that $\optoffline \leq \optoffline^{B}$. Therefore,
\begin{equation*}
   {\sf ALG_f} \leq  \gidN + \sum_{i=1}^{\gidN}{W(i)} \leq 1.7 \optoffline^{B} + \optoffline \leq 2.7 \optoffline^{B}\\
\end{equation*}
Hence, algorithm $\ourAlgoFirst$ when the first-fit strategy is used to compute $\binN$ is $2.7$-competitive.
\end{proof}

\section{Online Variable-size Drone Scheduling}
\label{sec:VariableSizeDS}
   In this section, we discuss a variant of the drone scheduling problem, and we call it \textsc{online variable-size drone scheduling} (OVDS). In the OVDS problem, we know all the customer requests in advance.
 So, we compute the underlying interval graph and partition the intervals into independent sets before the drones appear online. Each interval is associated with a cost reflecting the drone's battery consumption for the corresponding customer request. Drones with different battery capacities appear online. We allocate as many requests as possible to the current drone from one of the pre-computed independent sets. The battery capacity of every drone is at least the maximum cost of an interval. The goal is to use a minimum number of drones.

Our motivation for the OVDS problem comes from the following scenario: there is an allocation unit for drones that allocates a drone as and when requested by a delivery unit. The delivery unit is not aware of the battery capacities of the drones in advance. 
Therefore, upon request, when a drone is allotted, depending upon the drone's battery capacity, the delivery unit aims to schedule as many requests as possible for that drone. 
The delivery unit keeps raising requests to the allocation unit until all customer requests are allocated to some drones.
 The goal is to minimize the number of requests raised by the delivery unit.
 
Let $\mathcal{I}$ be the set of intervals corresponding to the customer requests. 
To allocate drones to the intervals, we assign a color to the intervals, similar to Section \ref{sec:DroneScehedulingBPC}. The color of an interval $I$ denoted by $\clr(I)$ corresponds to the drone allotted to serve the corresponding request. The $\clr(I)$ is defined as two tuples $\clr(I) = (\idN(I), \binN(I))$. We can compute the $\idN$ of any interval $I$ by using Algorithm \ref{alg:BPC-computeidNumber} and the set of intervals with a particular $\idN$ forms an independent set. Let $\{1, 2, \cdots, g\}$ be the set of $\idN$. There is a list of intervals corresponding to each of the $\idN$ in $\{1, 2, \cdots, g\}$ and every list is an independent set.
For each such list, we use the \textit{any-fit} strategy from the bin-packing problem and compute the $\binN$ for the intervals in that list. We present the algorithm $\VarBinDS$ that computes a color (allocates a drone) for each interval. 

\begin{definition}
\label{def:any-fit}
({\bf any-fit})
A bin packing algorithm is called any-fit algorithm if, after packing a bin, there is no item left in the list that could still fit into the bin.
\end{definition}
\begin{algorithm2e}[h]\normalsize
        \DontPrintSemicolon
	\caption{$\VarBinDS (\mathcal{I})$}	
	\label{alg:VariableSize}
 \scriptsize
	\For{each $I \in \mathcal{I}$}{
        compute-idNumber$(I)$\;
        }
        $g$ is the total $\idN$ used\;
        \For{$i = 1$ to $g$}{
        Create a list $List_i$ and $ List_i \gets \emptyset$\;
        }
        \For{each $I \in \mathcal{I}$}{
        Insert $I$ into $List_{\idN(I)}$\;
        }
        \For{each $i = 1$ to $g$}{
        Sort the intervals in $List_i$ as per the non-decreasing order of their costs\;
        }
        \For{each $i = 1$ to $g$}{
        $L_i \leftarrow 0$\;
        \While{$List_i$ is non-empty}{
        Request for a new drone\\
        $b$ is the capacity of this newly arrived drone\;
        $L_i \leftarrow L_i+1$; 
        $rem \leftarrow b$\;
        \For{each $I \in List_i$ as per the sorted order}{
        \If{$\C(I) \leq rem$}{
       $\binN(I) \leftarrow L_i$\\ 
       $rem \leftarrow$ $rem - \C(I)$\;
       Delete $I$ from $List_i$\;
        }
        \Else{
        break
        }
        }
        }
        }
\end{algorithm2e}
\begin{lemma}
    The running time of $\VarBinDS$ is $O(n \log n)$, where $n$ is the total number of requests or delivery.
    \label{lem:RunTime-VariablesizeDS}
\end{lemma}
\begin{proof}
    We can compute the $\idN$ of each interval in $O(n \log n)$ time ( follows from Lemma \ref{lem:runningtime-of-computerIdN}). Creating a list for each $\idN$ and assigning the intervals takes linear time. Thereafter sorting the intervals in each list takes $O(n \log n)$ time. At last, finding $\binN$ for each interval is done by some constant comparisons and assignments, which takes linear time again. Thus, the overall time complexity is $O(n \log n)$ time.\newline
\end{proof}
\begin{theorem}
    $\VarBinDS$ is $(2\alpha+1)$-competitive, where $\alpha$ is the ratio between the maximum and minimum battery capacity.
\end{theorem}
\begin{proof}
    Let $L_i$ be the number of different bins or $\binN$ created for $\idN = i (1 \leq i \leq g)$. Thus, $\{(i, 1), (i, 2), \cdots, (i, L_i)\}$
are the set of different colors used for the $\idN = i$. For each color, there is a request for a new drone.  So, there is an associated battery capacity for each color. Let $B_i^j$ be the battery capacity of the color $(i, j)$. Let $\C(i, j)$ be the sum of all the costs of the intervals having color $(i, j)$. Then, $\C(i, j) + \C(i, j+1) > B_i^j$, for $1 \leq i \leq g$ and $1 \leq j \leq L_i-1$. So, $\C(i, 1) + 2 \cdot \sum_{j=2}^{L_i-1}\C(i, j) + \C(i, L_i) >\sum_{j=1}^{L_i-1}B_i^j$ , implies $\sum_{i=1}^g \sum_{j=1}^{L_i-1}B_i^j < 2 \cdot \sum_{i=1}^g \sum_{j=1}^{L_i}\C(i, j)$. 

   Let $B_{min} = \underset{1 \leq i \leq g; 1\leq j \leq L_i}{\min} B_i^j$, and $B_{max} = \underset{1 \leq i \leq g; 1\leq j \leq L_i}{\max} B_i^j$. Then, $\alpha = \frac{B_{max}}{B_{min}}$. Let $OPT$ be the optimum colors used for the problem in the offline version when we have the knowledge of all the intervals along with their costs as well as all the drones with different battery capacities. Then, $OPT \geq \frac{\sum_{i=1}^g \sum_{j=1}^{L_i}\C(i, j)}{B_{max}}$. Also, $OPT \geq g$, from Lemma \ref{lem:gidNLeqOpt}.  Thus, the total number of colors used by the algorithm is $\sum_{i=1}^g L_i = \sum_{i=1}^g (L_i - 1) + g \leq \sum_{i=1}^g \sum_{j=1}^{L_i-1} (\frac{B_i^j}{B_{min}}) + g$. Therefore, 
   \setlength{\belowdisplayskip}{0.5pt} \setlength{\belowdisplayshortskip}{0.5pt}
\setlength{\abovedisplayskip}{0.5pt} \setlength{\abovedisplayshortskip}{0.5pt} 
 \begin{equation*}
     \sum_{i=1}^g L_i \leq \frac{2 \cdot \sum \limits_{i=1} \limits^g \sum \limits_{j=1}\limits^{L_i}\C(i, j)}{B_{min}} + g \leq \frac{2 \cdot \alpha \cdot  \sum \limits_{i=1}\limits^g \sum \limits_{j=1}\limits^{L_i}\C(i, j)}{B_{max}} + g \leq (2\alpha+1)OPT.\\
 \end{equation*} 
 Hence the proof.
\end{proof}

\section{Interval generator}
\label{sec:IntervalGen}
In Section \ref{sec:DroneScehedulingBPC}, we studied the \textsc{online drone Scheduling} problem where intervals appear online with a cost and every interval corresponds to a request of a customer.
Here, we analyze further how these intervals can be generated. For that reason, we present and study the \textsc{Interval Generator} problem.
Let us consider the following scenario. There is a truck that carries some drones and moves on a specific route throughout the day. Moreover, on the route have been predefined stops for the truck where the drones can take off or/and land.
During the movement of the truck, we receive a request from a customer that we want to serve. 
The drones will be used to serve the customers.
The goal here is to find the stops that a drone needs to take off and land that minimize a cost function.
Since in the drone scheduling problem for every interval (request) we want to assign a unique drone, 
we need to represent the takeoff and landing time for a particular request as an interval.
Thus, the \textsc{interval generator} problem receives a request in a specific time and produces an interval for this request with a cost based on the stops and the position of the truck.

More specifically, a request $q$ is given as coordinates in the plane and denoted by $q=(x_q,y_q)$.
Now, for the truck we need to consider two things, its route and its position on that.
The route can be considered as known and contains some stops.
Let $\mathcal{S}=\{S_1,S_2,\cdots,S_k\}$ be the set of stops on the route. We assume that the stops $S_1,\cdots,S_k$ are increasingly ordered w.r.t. the time of visit by the truck. Thus, the stop $S_1$ is the start point and $S_k$ is the end point of the truck. Every stop is given as coordinates in the plane, $S_i=(x_i,y_i)$ for all $i \in [1,k]$. Additionally, the stops will be used for the takeoff and landing of the drones as well as for the loading of the drones with the shipment. 
Moreover, the position $p$ of the track corresponds to some coordinates $(x_p,y_p)$ on the plane.

Next, the cost of a request can be considered as the time that a drone needs to serve a customer and to return to the truck (as we assume that the battery consumption is a linear function of the time).
For simplicity, we can assume that there is no battery consumption during the loading of the drones.
Furthermore, we say that the cost of a request is \emph{invalid} if the landing stop of the drone is `smaller' than the position of the truck at the time when the drone returns. Otherwise, we say that the cost of a request is \emph{valid}. We define it formally later.
Now we can construct an interval for the drone scheduling problem.
In Figure \ref{fig:example}, we can see an example of the above scenario where the route is a straight line.
Formally, we have the following definition for the problem.
\begin{definition}
\textbf{(\textsc{Interval Generator})}\\
\underline{\textbf{Input:}} A set of stops $\mathcal{S}=\{S_1,S_2,\cdots,S_k\}$, where each $S_i=(x_i,y_i)$, the truck position $p=(x_p,y_p)$, a request $q=(x_q,y_q)$ and a cost function $\C(\cdot)$.\\
\underline{\textbf{Output:}} An interval $I_{q}=[\ell_{q},r_{q}]$ for the interval graph $\mathcal{I}$ with the minimum valid cost.
\end{definition}

Observe that in the \textsc{Interval Generator} problem there is no restriction on the speed of the truck and the drones. Thus, we provide some assumptions to simplify the study of the problem.

\noindent \textbf{Assumptions}
(1) The speed of the drones and the truck is known and fixed. 
(2) The time of takeoff and landing as well as the time of loading and unloading is considered negligible. 
(3) The time that the truck remains on a stop is also considered negligible.
(4) The timeline of the day is known.

\begin{figure}[]
    \centering
    \includegraphics[width=9cm]{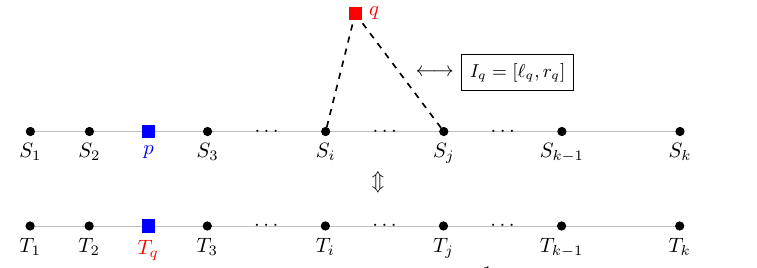}
    \caption{\textbf{Above}: The route of the truck is a straight line. The stops are given as black solid dots, the position $p$ of the truck is given as a \textcolor{blue}{blue square} on the route and the request $q$ as a \textcolor{red}{red square}. Moreover, $S_i$ is the takeoff stop and $S_j$ is the landing stop and $I_{q}=[\ell_{q},r_{q}]$ is the corresponding interval of the solution.
    \textbf{Below}: The timeline of the day with the corresponding times $T_1,\cdots,T_k$ of the stops. $T_q$ is the position of the truck in the timeline when we receive the request $q$.}
    \label{fig:example}
\end{figure}

 We describe now the way that we create an interval for a request for the \textsc{Interval Generator} problem. This can be done in two steps. First, we find the stops that minimize the cost of the request and after we compute the corresponding interval.
For the first part, by the input of the problem we have a set of stops $S$ and based on the assumptions we know the location of the truck every moment through the day.
Hence, we can correspond every stop $S_i$ to a specific time $T_i$ in the timeline of the day.
Moreover, when we receive a request $q$ we know the position $p$ of the truck and thus we can correspond the position of the truck as a time $T_q$ in the timeline in order to know when we have received the request. Next, we need to find the stops (times) for 
the takeoff and landing that minimize the cost of the request. 


As mentioned earlier the cost of a request is the total time from the takeoff stop to the position of the request and from there to the landing stop. 
Assume $S_i$ is the takeoff stop and $S_j$ is the landing stop; then we denote the total time of the delivery as $(S_i,q,S_j)$. 
Thus, in order to find the minimum cost we need to try all possible pairs of stops with $i<j$.
Also, the cost of the request needs to be valid. This means that the time of $(S_i,q,S_j)$ is not greater than the time that the truck needs to go from the stop $S_i$ to the stop $S_j$. Let $A$ be the set of all pairs that give valid cost.
Hence, the formula that computes the minimum valid cost for a request $q$ is the following: 
  $\C(q) = \min_{p \leq i <j,\, (i,j)\in A } (S_i,q,S_j).$

So far we have found the optimal stops for the delivery. Let $S_i$ and $S_j$ be the takeoff and landing stop respectively. We have corresponded every stop to a specific time in the timeline of the day. Thus, we know the takeoff time $T_i$ and the landing time $T_j$.
Now we are ready to create an interval $I_q=[\ell_{q},r_{q}]$ with left endpoint $\ell_{q}$, right endpoint $r_{q}$ and to assign a cost to the interval. 
By the definition of the interval graphs, every vertex corresponds to an interval of the real line in such a way that two vertices are adjacent if and only if the corresponding intervals have a nonempty intersection.
We can consider the timeline of the day as the real line and the times $T_i$ as points on the real line.
Then every valid delivery corresponds to an interval for the interval graph with left endpoint $\ell_{q}= T_i$ and right endpoint $r_{q}=T_j$. Moreover, we associate the cost of the request $q$ to this interval $I_q=[\ell_{q},r_{q}]$, that is $\C(I_q)=\C(q)$.
Hence, we have described how we can find a solution for the \textsc{Interval Generator} problem.
The Algorithm \ref{alg:intervalgenerator} computes an optimal solution for the \textsc{Interval Generator} problem.
\begin{algorithm2e}[h]
	\caption{$\IntervalGen$}	
        \DontPrintSemicolon
	\label{alg:intervalgenerator}	
 \footnotesize
	\textbf{Input:} Stops $\mathcal{S}=\{S_1,S_2,\cdots,S_k\}$, truck position $p=(x_p,y_p)$, request $q=(x_q,y_q)$ and cost function $\C(\cdot)$\;
		\textbf{Output:} An interval $I_{q}=[\ell_{q},r_{q}]$ with the minimum valid cost.\;
		\smallskip
		Let $M$ be the set of the indexes of the stops that are $\geq p$.\;
		\For{ $i \in M$}{
		\For{ $j \in M \quad and \quad i<j$}{
			\If{ $(i,j)$ is a valid pair}{		
			\If{$\C(q)>(S_i,q,S_j)$}{
				$\C(q)=(S_i,q,S_j)$ and
				$(i',j') = (i,j)$\;
			}
			}
		}
		}		
		 Find the times $T_{i'}$ and $T_{j'}$\; 
		Set $\ell_{q}= T_{i'}$ and  $r_{q}=T_{j'}$\;
		Set $\C(I_q)=\C(q)$\;
		\textbf{Return} $I_{q}=[\ell_{q},r_{q}]$\;
	
\end{algorithm2e}
\begin{theorem}
The running time of the Algorithm \ref{alg:intervalgenerator} is $O(k^2 \cdot \mathcal{T})$, where $\mathcal{T}$ is the time needed to compute each of $(S_i,q,S_j)$.
\label{thm:IntervalGenerator}
\end{theorem}
\begin{proof}
    We know that the set $M$ contains at most $k$ elements. Since we have two loops over this set, we have at most $k^2$ iterations.
    In each iteration, we need first to compute the cost $(S_i,q,S_j)$. This can be done in $\mathcal{T}$ time. 
    Next, we check if the cost is valid and if it is we compare it with the previous minimum cost. This can be done in constant time since the time that the truck needs to move from the stop $S_i$ to the stop $S_j$ is known by the assumptions.
    Thus, the running time of the Algorithm  \ref{alg:intervalgenerator} is $O(k^2 \cdot \mathcal{T})$.
\end{proof}
In Section \ref{sec:DroneScehedulingBPC}, we assume that all intervals have distinct endpoints or boundary points. However, Algorithm \ref{alg:intervalgenerator} can generate multiple intervals with common boundary points. 
To make each boundary point distinct, a small $\epsilon>0$ can be used
for such intervals. For example $I_{q}=[\ell_{q},r_{q}]$ and $I_{p}=[\ell_{p},r_{p}]$ with one common boundary point, i.e., $r_{q}= \ell_{p}$. These can be transformed into $[\ell_{q},r_{q}-\epsilon]$ and $[\ell_{p} + \epsilon,r_{p}]$ to ensure all boundary points of the $I_{p}$ and $I_{q}$ are distinct.
\section{Conclusion}
\label{sec:conclusion}
In this paper, we have studied a hybrid truck-drone model for last-mile delivery where delivery requests appear online during the moving truck carrying multiple drones in a predefined path. A drone is enabled to launch from the truck to deliver a parcel to the customer, as per the request, and after delivery, it returns and lands on the truck.
The objective is to minimize the number of drones allocated to complete all feasible deliveries. To address this, we have proposed two algorithms for \textsc{online drone scheduling} problem based on next-fit and first-fit bin-packing strategies. We have analyzed the performances of the algorithms in terms of competitive ratio and worst-case time complexity. We also have introduced \textsc{online variable-size drone scheduling} problem and proposed an algorithm to solve it. We have discussed the interval generator, a tool to calculate delivery time intervals for each request. The interval generator aims to minimize the delivery cost associated with each request, choosing delivery intervals from predefined discrete stopping points on the truck's path. We can treat these points as mini-warehouses where packages can be loaded into the truck, if necessary, to fulfill the requests. 
For future work, evaluating the optimal drone-truck path to deliver all types of packages would be interesting along with a simulation study for possible practical application. Another interesting problem to address in the future is determining the delivery interval for each request when there are no predefined discrete stopping points along the truck's path. The model's applicability can also be extended to a drone-delivery system with a recharging battery policy for the drones.

\bibliographystyle{splncs03}
\bibliography{Refer}

\end{document}